\pdfoutput=1
\RequirePackage{ifpdf}
\ifpdf 
\documentclass[pdftex]{sigma}
\else
\documentclass{sigma}
\fi

\usepackage{shuffle}
\usepackage{tikz}
\usepackage{slashed}
\usepackage{adjustbox}
\usetikzlibrary{shapes.geometric, arrows}

\newcommand{\ol}{\overline}

\newcommand{\cM}{\mathcal{M}}

\newcommand{\cL}{\mathcal{L}}

\newcommand{\fS}{\mathfrak{S}}

\newcommand{\fsu}{\mathfrak{su}}

\newcommand{\bR}{\mathbb{R}}
\newcommand{\bC}{\mathbb{C}}

\newcommand{\bP}{\mathbb{P}}

\newcommand{\tr}{\operatorname{tr}}
\newcommand{\Lie}{{\mathrm{Lie}}}

\newcommand{\YM}{\mathrm{YM}}
\newcommand{\pt}{\mathrm{pt}}
\newcommand{\PT}{\mathrm{PT}}

\newcommand{\Conf}{\mathrm{Conf}}

\newcommand{\Sh}{\mathrm{Sh}}

\newcommand{\Thin}{\mathrm{Thin}}
\newcommand{\Res}{\mathrm{Res}}

\newcommand{\PSL}{\mathrm{PSL}}
\newcommand{\SL}{\mathrm{SL}}
\newcommand{\GL}{\mathrm{GL}}

\renewcommand{\d}{\text{d}}
\newcommand{\pa}{\partial}

\newcommand{\avg}[1]{\left< #1 \right>} 	

\numberwithin{equation}{section}

\newtheorem{Theorem}{Theorem}[section]

\newtheorem{Lemma}[Theorem]{Lemma}
\newtheorem{Proposition}[Theorem]{Proposition}
 { \theoremstyle{definition}

\newtheorem{Remark}[Theorem]{Remark} }

\begin{document}
\allowdisplaybreaks

\renewcommand{\thefootnote}{}

\renewcommand{\PaperNumber}{101}

\FirstPageHeading

\ShortArticleName{The Algebraic Structure of the KLT Relations for Gauge and Gravity Tree Amplitudes}

\ArticleName{The Algebraic Structure of the KLT Relations\\ for Gauge and Gravity Tree Amplitudes\footnote{This paper is a~contribution to the Special Issue on Algebraic Structures in Perturbative Quantum Field Theory in honor of Dirk Kreimer for his 60th birthday. The~full collection is available at \href{https://www.emis.de/journals/SIGMA/Kreimer.html}{https://www.emis.de/journals/SIGMA/Kreimer.html}}}

\Author{Hadleigh FROST}

\AuthorNameForHeading{H.~Frost}

\Address{Mathematical Institute, University of Oxford, Oxford, UK}
\Email{\href{mailto:frost@maths.ox.ac.uk}{frost@maths.ox.ac.uk}}

\ArticleDates{Received March 01, 2021, in final form November 01, 2021; Published online November 14, 2021}

\Abstract{We study the Kawai--Lewellen--Tye (KLT) relations for quantum field theory by reformulating it as an isomorphism between two Lie algebras. We also show how explicit formulas for KLT relations arise when studying rational functions on ${\mathcal M}_{0,n}$, and prove identities that allow for arbitrary rational functions to be expanded in any given basis. Via the Cachazo--He--Yuan formulas for, these identities also lead to new formulas for gauge and gravity tree amplitudes, including formulas for so-called Bern--Carrasco--Johansson numerators, in the case of non-linear sigma model and maximal-helicity-violating Yang--Mills amplitudes.}


\Classification{05C05; 17B62; 81T13; 81T30}

\renewcommand{\thefootnote}{\arabic{footnote}}
\setcounter{footnote}{0}

\section{Introduction}\label{sec:1}
The \emph{Kawai--Lewellen--Tye $($KLT$)$ relations} for field theory amplitudes express the tree amplitude of a gravity theory as a quadratic expression in the tree amplitudes of a gauge theory. The existence of such a relation is motivated by the \emph{string KLT relation} between open and closed string tree amplitudes. These relations were proposed in \cite{1985KLT}, and explicit formulas for the components of the KLT matrix were derived more recently in \cite{2010BBDSV}. The KLT relation for string amplitudes resembles the Riemann period relations, because the quadratic relation is obtained by taking the inverse of a matrix of the intersection numbers (in a certain homology theory) of~$\cM_{0,n}$, the moduli space of~$n$ points in~$\bP^1$. The KLT relations are studied from this point of view in~\cite{2018BD,2017Miz}. There has been work done to understand the field theory KLT relation in similar terms \cite{2018Miz,2019Miz}.

This article studies the field theory KLT relation directly, without reference to $\cM_{0,n}$, to find a~simple algebraic interpretation of the KLT matrix. The relationship to $\cM_{0,n}$ is then exploited to find new formulas for some gauge theory tree amplitudes. The formulas obtained are in a~special form, which manifests the numerators of Bern--Carrasco--Johansson (BCJ) \cite{2008BCJ}. The rest of this introduction gives a brief review of the field theory KLT relations, and then summarizes the results of the paper.

\subsection{Partial amplitudes and the field theory KLT relations}
The \emph{colour factors} of a gauge theory amplitude organize the amplitude into \emph{partial amplitudes} corresponding to surfaces with boundary marked points. For an ${\rm SU}(N)$ gauge theory, with coupling constant $g_\YM$, the full amplitude may be written as
\begin{gather}\label{eqn:p2}
A_{n} = ({\rm i}g_\YM)^{n-2}\sum_{\ell=0}^{\infty} \lambda^\ell \sum_{h,p,g} \left(\frac{1}{N}\right)^{2g+h-1} A_{p,g,h},
\end{gather}
where $\lambda = g_\YM^2N$ is the 't Hooft coupling \cite{1973H}, and where $A_{p,g,h}$ is the sum of partial amplitudes corresponding to surfaces of genus $g$, with $p$ punctures and $h$ boundary components. The second sum in~\eqref{eqn:p2} is constrained by the relation $p+2g+h = \ell+1$.

Consider Yang--Mills (YM) gauge theory, with gluons in the matrix representation of $\fsu(N)$. Fix some distinct $\lambda_1,\dots ,\lambda_n \in \fsu(N)$, corresponding to the external gluon colour states. A cubic Feynman tree diagram, $\alpha$, gives a contribution to the tree amplitude which is proportional to
\begin{gather}\label{eq:colourfactordef}
c_\alpha := \tr(\alpha[\lambda_1,\dots,\lambda_{n-1}] \lambda_n),
\end{gather}
where $\tr$ is the Killing form, and $\alpha[\lambda_1,\dots,\lambda_{n-1}]$ is the Lie bracketing of the $\lambda_i$ according to the tree $\alpha$, regarded as a rooted binary tree, with root $n$.\footnote{A binary rooted tree with $k$ leaves determines and is determined by a Lie bracketing of~$k$ variables, up to a~sign.} A Feynman tree diagram that contains quartic vertices gives a contribution which is a sum of terms, each proportional to $c_\alpha$ for some binary rooted tree~$\alpha$.\footnote{The cubic diagrams that appear in this sum can be found by considering all the possible ways to `expand' the quartic vertices into a subdiagram with two cubic vertices.} It is therefore possible to write the $n$-point tree amplitude as
\begin{gather*}
A_\YM^{\text{tree}} = \sum_{\substack{\text{trees} \\ \alpha}} A_\alpha c_\alpha,
\end{gather*}
where the sum is over all binary trees $\alpha$, with external edges labelled by $1, \dots , n$. The coefficients~$A_\alpha$ that appear in the sum depend only on the gluon momenta and polarizations. The colour factors $c_\alpha$ may be expanded as a sum over permutations:
\begin{gather}\label{eq:CA}
c_\alpha = \sum_{a\in\fS_{n-1}} (a,\alpha) \tr (\lambda_{a(1)}\dots \lambda_{a(n-1)}\lambda_{n}),
\end{gather}
where $\fS_{n-1}$ is the set of permutations. The bracketing $(a,\alpha)$ denotes the coefficient of the ordering $a$ in the expansion of the monomial $\alpha$, which can be either $+1$, $-1$, or $0$. The partial amplitude expansion of $A_\YM^{\text{tree}}$ can therefore also be written as a sum over permutations:
\begin{gather*}
A_\YM^{\text{tree}} = \sum_{a\in\fS_{n-1}} A(a,n) \tr (\lambda_{a(1)}\dots \lambda_{a(n-1)}\lambda_{n}),
\end{gather*}
where $A(a,n)$ is the sum over binary trees,
\begin{gather*}
A(a,n) = \sum_{\substack{\text{trees} \\ \alpha}} (a,\alpha) A_\alpha.
\end{gather*}
Given the gauge theory partial tree amplitudes, the \emph{field theory KLT relation} expresses the $n$-point gravity tree amplitude~$M_n$ as a quadratic expression of the following form \cite{2010BBDSV,2009BBDV}
\begin{gather}\label{eq:kltrelation}
M_n = \lim_{k_n^2\rightarrow 0} \sum_{a,b \in \fS_{n-2}} \frac{A(1an) S(1a,1b) A\big(\overline{b}1n\big)}{k_n^2},
\end{gather}
where the matrix entries $S(1a,1b)$ depend only on the gluon momenta, $k_i^\mu$. It is understood that, $S(1a,1b)$ is defined `off-shell', in the sense that it is valid for $k_n^2 \neq 0$. When $k_n^2 \neq 0$, $S(1a,1b)$ is a $(n-2)!\times (n-2)!$ matrix of full rank (as explained in Section~\ref{sec:2}). Its rank drops by $1$ when $k_n^2 = 0$. The explicit formula for the entries of $S(1a,1b)$ is given by~\cite{2010BBDSV}
\begin{gather}\label{eq:kltkernel}
S(1a,1b) := \prod_{i=2}^{n-1} \left( \sum_{\substack{j<_{1a} i\\ j>_{1b} i}} s_{ij}\right),
\end{gather}
where the variables
\begin{gather*}
s_{ij} = 2 k_i \cdot k_j
\end{gather*}
are the Mandelstam variables associated to the gluon momenta. For fixed $i$, the sum in this formula is over all $j$ that both precede $i$ in $1a$, and are preceded by $i$ in $1b$. The notation $\overline b$ denotes the reversal of the ordering~$b$.

This formula, \eqref{eq:kltrelation}, was originally derived in \cite{2010BBDSV} using an argument from the properties of the open string integral. \cite{2020FrostMM} showed that the `off-shell' KLT matrix, \eqref{eq:kltkernel}, is the matrix inverse of a~matrix of `Berends--Giele' currents for bi-adjoint scalar theory. Namely, define a~$(n-2)!\times (n-2)!$ matrix, $T$:
\begin{gather*}
T(1a,1b) = \sum_{\substack{\text{trees} \\ \alpha}} \frac{(1a,\alpha)(1b,\alpha)}{k_n^2 s_\alpha},
\end{gather*}
where $a$ and $b$ are permutations of $2, \dots , n-1$, and the sum is over all binary trees. The brackets $(1a,\alpha)$ are defined as in~\eqref{eq:CA}, above. The denominator, $s_\alpha$, is the \emph{product of propagators}.\footnote{$s_\alpha$ is a polynomial expression in Mandelstam variables, determined by the tree $\alpha$; see~\eqref{eq:salpha}.} Then, thinking of $S(1a,1b)$ as an $(n-2)!\times (n-2)!$ matrix, $S$:
\begin{Proposition}
The matrices $S$ and $T$ are inverse: $ST = TS = {\rm Id}$, i.e.,
\begin{gather}\label{eq:introST}
\sum_{b \in \fS_{n-2}} T(1a,1b) S(1b, 1c) = \sum_{b \in \fS_{n-2}} S(1a,1b) T(1b, 1c) = \delta_{b,c},
\end{gather}
where $\delta_{b,c}$ is the identity matrix.
\end{Proposition}

Section~\ref{sec:2} gives a basis-independent statement of this result, and gives a streamlined version of the proof in~\cite{2020FrostMM}.

\subsection{Summary}
This section summarizes the new results in Sections~\ref{sec:3} and~\ref{sec:4}, which build on the approach taken in Section~\ref{sec:2}, and culminate in formulas for some gauge and gravity tree amplitudes. The proof of~\eqref{eq:introST} involves a bracket operation (originally called the `S-map' in \cite{2014MS,2015MS}) defined on orderings. Some examples are
\begin{gather}\label{eq:introex}
\{1,23\} = s_{12} 123 - s_{13} 132, \qquad \{12,34\} = s_{23} 1234 - s_{13} 2134 + s_{14} 2143 - s_{24} 1243.
\end{gather}
The definition of $\{a,b\}$ is given in Section \ref{sec:2}, where it is also shown that $\{~,~\}$ is a Lie bracket. As explained in Section~\ref{sec:3}, and exploited in Section~\ref{sec:4}, the bracket operation $\{a,b\}$ naturally arises in the context of rational functions on the configuration space, $\Conf_{n-1}(\bC)$, of points in the complex plane. Write $\pt(12\dots n)$ for the function
\begin{gather*}
\pt(12\dots n) = \prod_{i=1}^{n-1} \frac{1}{z_i - z_{i+1}}.
\end{gather*}
Then (proved in Section~\ref{sec:3})
\begin{Proposition}
For an ordering $a$ of $1,\dots,k$, and an ordering $b$ of $k+1,\dots,n$:
\begin{gather}\label{eq:introSPT}
\left( \sum_{i=1}^k \sum_{j=k+1}^n \frac{s_{ij}}{z_i - z_j} \right) \pt(a) \pt(b) = \pt(\{a,b\}).
\end{gather}
\end{Proposition}

When combined with the definition of $S(1a,1b)$, \eqref{eq:introSPT} can be used to derive a formula for the partial tree amplitudes of the non-linear sigma model (the gauge theory associated to maps from the Riemann sphere to the Lie group ${\rm SU}(N)$). The result is
\begin{Proposition}The NLSM tree partial amplitudes are given by
\begin{gather}\label{intro:ANLSM}
A_{{\rm NLSM}} (a,n) = \sum_{b\in\fS_{n-2}} S(1b,1b) m(1b,n|a,n),
\end{gather}
where the sum is over all permutations of $2,3,\dots,n-1$, and $m(a,n|b,n) = k_n^2 T(1a,1b)$ are the partial biadjoint scalar amplitudes.
\end{Proposition}

This agrees with earlier results reported in \cite{2017CMS,2020Mafra}; but the methods used to derive it are new. The idea that leads to~\eqref{intro:ANLSM} is to use the matrix tree theorem to expand the matrix determinant that appears in the integrand of the so-called CHY formula for $A_{\text{NLSM}}$. Identities proved in Section~\ref{sec:3} are then used to re-arrange the integrand into a suitable form. The most useful identity is
\begin{Proposition}
Let $G$ be a tree with vertex set $1,\dots,n$. Orient the edges of $G$ by fixing~$1$ to be a sink, and demanding that all vertices $($except for~$1)$ have only one outgoing edge. For $i\neq 1$, write $x(i)$ for the endpoint of the edge outgoing from $i$. Then
\begin{gather}\label{intro:Trat}
\prod_{\substack{{\rm edges} \\ i\rightarrow j}} \frac{1}{z_{ij}} = \sum_{\substack{a\\ x(i)<_{1a}i}} \pt(1a).
\end{gather}
The sum is over all orderings, $a$, such that $x(i)$ precedes $i$ in $1a$, for all $i$.
\end{Proposition}

The identity, \eqref{intro:Trat}, can also be used to obtain formulas for other tree amplitudes that have CHY formulas (or similar). In four dimensions, gravity tree amplitudes can be expressed in terms of determinants of matrices called Hodges' matrices. Applying the same idea that leads to~\eqref{intro:ANLSM} also leads to new formulas for 4D gravity (and then, by the KLT relations, for Yang--Mills) amplitudes. These formulas are not always easy to evaluate. However, in the MHV case (when only 2 gluons are $+$ helicity, and the rest are $-$ helicity), contact can be made with known results. The MHV result, expressed in the standard spinorial notation used, is
\begin{Proposition}The tree level gravity amplitude can be expanded as
\begin{gather}\label{intro:newG3}
M_{\rm GR} = \sum_{\sigma \in \fS_{n-3}} A_{\rm YM}(12b\sigma) N(12b\sigma),
\end{gather}
where
\begin{gather*}
N(12b\sigma) = \frac{\avg{12}[12]^2}{[b1][b2]} \prod_{\substack{j=3 \\ i\neq b}}^n \sum_{i<_\sigma j} \frac{[i2][ij]}{[j2]}.
\end{gather*}
\end{Proposition}
This is closely related to the early result due to Berends--Giele--Kuijf \cite{1988BGK}, which gives a formula for $M_{\text{GR}}$, but in a different form. The methods used in the derivation lead directly to expressions of the form of~\eqref{intro:newG3}, which is well suited to the KLT relations. Indeed, the field theory KLT relations imply that, given~\eqref{intro:newG3}, the YM partial amplitudes can be written as
\begin{gather*}
A_{\rm YM}(\rho) = \sum_{\sigma \in \fS_{n-3}} m(12b\sigma|\rho) N(12b\sigma).
\end{gather*}
Finally, the discussion in Section~\ref{sec:5} concludes by relating the approach taken in this paper to three important outstanding problems for our understanding of perturbative gauge theory amplitudes and the KLT relations.

\section{The KLT kernel}\label{sec:2}
The KLT kernel, \eqref{eq:kltkernel}, is the inverse of a certain map, as shown in~\cite{2020FrostMM}. This section revisits this result, emphasizing those aspects which are relevant for Sections~\ref{sec:3} and~\ref{sec:4}.

To fix notation, write $A$ for the set $\{1,\dots,n\}$. An \emph{ordering} of $A$ is a word that uses each letter $i\in A$ exactly once. Write $\fS_A$ for the set of orderings of $A$, and $W_A$ for the $\bR$-vector space generated by $\fS_A$. There is a multilinear inner product, $(~,~)$, on $W_A$, such that, for two distinct orderings $a$ and $b$,
\begin{gather*}
(a,a) = 1,\qquad\text{and}\qquad (a,b) = 0.
\end{gather*}
Let $\cL_A \subset W_A$ be the subspace of \emph{multilinear Lie polynomials} in~$W_A$.\footnote{If $\Lie(A)$ is the free Lie algebra on $A$, then $\cL_A$ is the intersection $\Lie(A) \cap W_A$, using the inclusions of $W_A$ and $\Lie(A)$ into the free associative algebra on~$A$.} And let $\Sh_A \subset W_A$ be the subspace of nontrivial shuffle products: $\Sh_A$ is linearly spanned by expressions of the form
\begin{gather*}
a \shuffle b,
\end{gather*}
where $a$ and $b$ are two (non-empty) words whose concatenation, $ab$, is an ordering of~$A$.\footnote{The shuffle product, $a\shuffle b$, is defined inductively by $(ia)\shuffle (jb) = i(a\shuffle jb) + j (ia \shuffle b)$, and $a\shuffle e = e \shuffle a = a$, where $i$, $j$ are individual letters, $e$ is the empty word, and $a$, $b$ are nonempty words.} By Ree's theorem \cite{1993R}, $\cL_A$ is the orthogonal subspace
\begin{gather*}
\cL_A = \Sh_A^\perp,
\end{gather*}
with respect to the given inner product. The dual vector space $\cL_A^\vee$ is
\begin{gather*}
\cL_A^\vee = W_A / \Sh_A.
\end{gather*}

\begin{Remark}\label{rmk:shA} An alternative definition of $\Sh_A$ is as follows. It is the subspace in $W_A$ spanned by the expressions
\begin{gather}\label{eq:kksha}
aib - (-1)^{|a|} i (\overline{a}\shuffle b),
\end{gather}
where $aib$ is an ordering in $\fS_A$, and $i\in A$ is a single letter, and $|a|$ denotes the length of the word~$a$. See \cite[Corollary 2.4]{2002S} or \cite[Lemma~3.6]{Thesis}.
\end{Remark}

A Lie monomial $\alpha \in \cL_A$ defines a \emph{rooted binary tree} with leaves labelled by $1,\dots,n$. See Figure~\ref{fig:tree} for an example. If the Lie monomial $\alpha$ is written in its bracketed form, then the associated tree has one internal edge for each pair of brackets in $\alpha$, and this edge can be labelled by the subset $I\subset A$ of letters that appear inside that pair of brackets. Write $P(\alpha)$ for the set of edges of $\alpha$ (including the root edge). For example,
\begin{gather*}
P([[1,2],3]) = \{\{12\},\{123\}\}.
\end{gather*}
Given null momenta $k_i^\mu$ for each $i=1,\dots,n$ \big($k_i^2 = 0$\big), form the associated Mandelstam variables
\begin{gather*}
s_I = \left( \sum_{i\in I} k_i^\mu \right)^2,
\end{gather*}
or, equivalently,
\begin{gather}\label{eq:linear}
s_I = \sum_{\{ij\}\subset I} s_{ij},
\end{gather}
where $s_{ij} = 2 k_i \cdot k_j$. Given a tree $\alpha$, associate to each edge, $I \in P(\alpha)$, a massless scalar propagator: $1/s_I$. For each rooted binary tree, $\alpha$, introduce the following `product of propagators' monomial
\begin{gather}\label{eq:salpha}
\tilde s_\alpha = \prod_{I\in P(\alpha)} s_I.
\end{gather}
Finally, it is convenient to write $M_A$ for the Laurent ring with variables $s_I$, $I\subset A$, subject to the linear relations above, \eqref{eq:linear}.

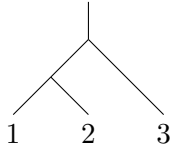
\begin{figure}[t]\centering
\begin{tikzpicture}[scale=0.5]
\draw (0,0) node[below]{1} -- (0.5,0.5);
\draw (0.5,0.5) -- (1,1);
\draw (1,1) -- (1.5,0.5);
\draw (1.5,0.5) -- (2,0) node[below]{2} ;
\draw (1,1) -- (1.5,1.5);
\draw (1.5,1.5) --(2,2);
\draw (2,2) -- (3,1);
\draw (3,1) --(4,0) node[below]{3};
\draw (2,2) -- (2,3);
\end{tikzpicture}
\caption{The labelled rooted binary tree corresponding to the Lie monomial $\pm [[1,2],3]$.}\label{fig:tree}
\end{figure}

For an ordering $a\in \fS(A)$, write
\begin{gather*}
T(a) = \sum_{\alpha} \frac{(a,\alpha)\alpha}{\tilde s_{\alpha}} \in \cL_A\otimes M_A,
\end{gather*}
where the sum is over all rooted binary trees with leaves labelled by $A$. $T(a)$ is a `prototype' of a gauge theory partial tree amplitude. Kapranov~\cite{2012Kap} proposed to study the KLT relation by regarding $T(a)$ as defining a linear map
\begin{gather*}
T \colon \ \cL_A^\vee\otimes M_A \rightarrow \cL_A\otimes M_A.
\end{gather*}
The map $T$ is self-adjoint with respect to the pairing between $\cL_A^\vee$ and $\cL_A$,
\begin{gather}\label{eq:Tself}
(a,T(b)) = (b,T(a)).
\end{gather}
Moreover, the functions $1/\tilde s_\alpha$ are linearly independent in $M_A$, and this implies that $\ker T$ is trivial. Indeed, for $P \in \cL_A^\vee$, if $T(P) = 0$, then $(P,\alpha) = 0$ for all Lie monomials $\alpha$. By dimension counting, $T$ is onto, so it follows that $T$ is an isomorphism.

Write $\tilde\cL_A = \oplus_{B\subset A} \cL_B$, and so on. $T$ extends to an isomorphism
\begin{gather*}
\tilde T \colon \ \tilde \cL_A^\vee\otimes M_A \rightarrow \tilde \cL_A \otimes M_A.
\end{gather*}
$\tilde\cL_A\otimes M_A$ is a Lie algebra with the usual bracket.\footnote{But note that the bracket of two Lie monomials must be zero in $\tilde \cL_A$ if they share any letter in common.} The isomorphism $\tilde T$ then induces a Lie bracket on the dual vector space! Indeed, define a bracket operation, $\{~,~\}$, on $\cL_A^\vee\otimes M_A$ by
\begin{gather}\label{eq:intertwine}
T(\{a,b\}) = [T(a),T(b)],
\end{gather}
for $a,b \in \cL_A^\vee$. Since $\tilde T$ is an isomorphism, $\{~,~\}$ is a Lie bracket. In fact, for disjoint orderings~$a$ and~$b$, it can be shown that, explicitly
\begin{gather}\label{eq:Sexplicit}
\{a,b\} = \sum_{\substack{a=a_1ia_2\\b=b_1jb_2}} (-1)^{|a_2|+|b_1|}s_{ij} (a_1\shuffle\overline{a_2})ij\big(\overline{b_1}\shuffle b_2\big).
\end{gather}
Examples of this were given above, in \eqref{eq:introex}. See also \cite[Section~4]{2020FrostMM} or \cite[Chapter~4]{Thesis}.

For a Lie monomial~$\alpha$, written in bracketed form, let $S(\alpha) \in \cL_A^\vee\otimes M_A$ be obtained from $\alpha$ by replacing every pair of brackets with a pair of braces. For example,
\begin{gather*}
S([[1,2],3]) = \{\{1,2\},3\}.
\end{gather*}
This extends to define a linear map
\begin{gather*}
S \colon \ \cL_A\otimes M_A \rightarrow \cL_A^\vee\otimes M_A.
\end{gather*}
Repeated applications of \eqref{eq:intertwine} gives that
\begin{gather*}
T(S(\alpha)) = \alpha.
\end{gather*}
This implies that $S$ is self-adjoint:
\begin{gather*}
(\beta,S(\alpha)) = (T(S(\beta)),S(\alpha)) =(S(\beta),T(S(\alpha))) = (S(\beta),\alpha),
\end{gather*}
using \eqref{eq:Tself}.

\begin{Proposition}\label{prop:TS}
$T$ and $S$ are inverses.
\end{Proposition}

\begin{proof}
Let $b_i$, $\beta_i$ be a pair of dual bases for $\cL_A$ and $\cL_A^\vee$ (with $i = 1,\dots, (n-1)!$). For an ordering $a\in W_A$,
\begin{gather*}
S(T(a)) = \sum_i S(\beta_i) (b_i,T(a)) = \sum_{i,j} b_j (S(\beta_i), \beta_j) (b_i,T(a)).
\end{gather*}
Using that $S$ is self-adjoint gives
\begin{gather*}
rhs = \sum_{i,j} b_j (\beta_i, S(\beta_j)) (b_i,T(a)) = \sum_j b_j (T(a), S(\beta_j)) = \sum_j b_j (a,\beta_j).
\end{gather*}
But $b_i$, $\beta_i$ is a pair of dual bases, so $S(T(a)) = a$.
\end{proof}

Fixing a pair of dual bases as above, the components of $T$ and $S$ are $(b_j,T(b_i)) = T_{ij}$ and $(\beta_j,S(\beta_i)) = S_{ij}$. These are $(|A|-1)!\times (|A|-1)!$ matrices, and the proposition says that~$S_{ij}$ is the matrix inverse of $T_{ij}$. This very simple definition of $S$ was missed in the literature, possibly because of the limit that appears in~\eqref{eq:kltrelation}. Moreover, notice that, trivially,
\begin{gather}\label{eq:TTST}
T_{il} = \sum_{j,k} T_{ij} S_{jk} T_{kl}.
\end{gather}
This will be seen to imply the field theory KLT relations for the gauge and gravity-like theories studied in Section~\ref{sec:4}.

A possible choice of dual bases for $\cL_A$ and $\cL_A^\vee$ is to take the $(n-1)!$ words $1b$, for each $b$ an ordering in $\fS(2,\dots,n)$, and dually the $(n-1)!$ Lie monomials
\begin{gather*}
\ell(1b) = [[\dots [[1,b(1)],b(2)],\dots ],b(n-1)] \in \cL_A,
\end{gather*}
for each $b \in \fS(2,\dots,n)$. These are dual bases because
\begin{gather*}
(1b,\ell(1b')) = \begin{cases} 1 & \text{if $b=b'$,} \\ 0 & \text{otherwise.} \end{cases}
\end{gather*}
In these bases, the components of $S$ are found to be \cite{2020FrostMM}
\begin{gather}\label{eq:Scomponent}
S(1a,1b) = (\ell(1b),S(\ell(1c))) = \prod_{i=2}^{n-1} \sum_{\substack{j<_{1b} i\\j<_{1c} i}} s_{ij}.
\end{gather}
This formula, discussed in the introduction, is the formula first found by~\cite{2010BBDSV} (albeit with the order of one of the words reversed). Many variations on this formula (including the original formula in~\cite{2010BBDSV}) can be obtained by choosing to compute the matrix elements of~$S$ using a different pair of bases.

\section{Scattering equations identities}\label{sec:3}
The `Cachazo--He--Yuan (CHY)' formulas express the partial tree amplitudes of several gauge theories as a sum of resides of logarithmic forms on $\cM_{0,n}$, as reviewed in Section~\ref{sec:43}. The logarithmic forms on $\cM_{0,n}$ satisfy algebraic identities that imply, via the CHY formulas, identities amoung the associated gauge theory tree amplitudes. This section concludes by proving one such identity, Proposition~\ref{prop:tree}, which is used in applications to gauge theory amplitudes in Section~\ref{sec:4}.

The open stratum of the moduli space $\cM_{0,n}$ is defined as
\begin{gather*}
\cM_{0,n}(\bC) = \big(\bC\bP^1\big)^{\oplus n}_* / \PSL_2\bC,
\end{gather*}
where $\big(\bC\bP^1\big)^{\oplus n}_*$ denotes $n$-tuples of pairwise distinct points in $\bC\bP^1$, and $\PSL_2\bC$ acts by M\"obius transformations. Write $\bC_{n-1}^*$ for the braid hyperplane arrangement, $\bC^{n-1}_*:=\bC^{n-1} - \Delta$, where~$\Delta$ is the big diagonal (i.e., the union of the hyperplanes $z_i-z_j = 0$). The open stratum of $\cM_{0,n}$ is the quotient of this by the free action of $\bC^*\ltimes \bC$, that acts as $(a,b)\colon z \mapsto az+b$,
\begin{gather*}
\cM_{0,n}(\bC) \simeq \bC^{n-1}_* / \bC^*\ltimes \bC.
\end{gather*}
This follows by setting $z_n = \infty$, and noticing that the stabalizer in $\PSL_2\bC$ of a point in $\bP^1$ is $\bC^*\ltimes \bC$.

Write $z_{ij} := z_i - z_j$. In the ring of rational functions on $\bC^{n-1}_*$, there is a natural submodule spanned by the \emph{broken Parke--Taylor} functions, $\pt(a)$, defined for a given word $a = 123\dots n-1$, as the product
\begin{gather*}
\pt(123\dots n-1) = \prod_{i=1}^{n-2} \frac{1}{z_{ii+1}}.
\end{gather*}
It can be shown (by an explicit induction, or by using general results from \cite{1999BV}) that these functions satisfy
\begin{gather}\label{eq:ptshuffle}
\pt(a\shuffle b) = 0,
\end{gather}
for any two disjoint (non-empty) words $a$ and $b$ (i.e., two words with no letters in common). In view of~\eqref{eq:kksha},~\eqref{eq:ptshuffle} further implies that
\begin{Lemma}\label{lem:ptkk}
For an ordering $a1b \in \fS(1,\dots,n-1)$,
\begin{gather*}
\pt(a1b) = (-1)^{|a|} \pt(1(\overline{a}\shuffle b)).
\end{gather*}
\end{Lemma}

This lemma implies that open string partial tree amplitudes satisfy the so-called Kleiss--Kuijf relations~\cite{1988KK}.

If $a$ and $b$ are two disjoint words, then for any $i\in a$ and $j\in b$, the Lemma implies that
\begin{gather*}
\frac{1}{z_{ij}} \pt(a)\pt(b) = (-1)^{|a_2|+|b_1|} \pt \big((a_1\shuffle\overline a_2) i j \big(\overline b_1 \shuffle b_2\big)\big),
\end{gather*}
where $a = a_1ia_2$ and $b = b_1jb_2$. Recalling \eqref{eq:Sexplicit}, this implies that
\begin{Lemma}\label{lem:ptS} For $a$ and $b$ disjoint words as above,
\begin{gather*}
\pt(\{a,b\}) = \pt(a)\pt(b) E_{a,b},
\end{gather*}
where
\begin{gather*}
E_{a,b} = \sum_{i\in a, j\in b} \frac{s_{ij}}{z_{ij}}.
\end{gather*}
Note that $i\neq j$ for each term in the sum.
\end{Lemma}

The equations $E_{a,b}=0$ are known as the \emph{scattering equations}.\footnote{Let $ab$ be an ordering of $1,\dots,n$, then the functions $E_{a,b}$ arise as derivatives of the \emph{Koba--Nielsen function},
\[
f_s(z)= \prod_{i=1,i<j}^n z_{ij}^{s_{ij}}.\]
 Indeed, \[
 E_{a,b} = \sum_{i \in a} \frac{\pa f_s(z)}{\pa z_i}.\]} Fix $A = \{1,\dots,n-1\}$, and take the ordering $a = 123\dots n-1$. Then repeated applications of Lemma~\ref{lem:ptS} give that
\begin{gather*}
\prod_{i=2}^{n-1} E_{i,123\dots i-1} = \sum_{b\in\fS(2\dots n-1)} S(12\dots n-1,1b) \pt(1b).
\end{gather*}
In this way, the components of the KLT matrix $S$ can be recovered from products of the func\-tions~$E_{a,b}$.

\begin{Remark}\label{rmk:bcj} Lemma \ref{lem:ptS} can also be used \cite{2012C,2014MS} to show that gauge theory partial tree amplitudes satisfy the \emph{fundamental BCJ relations} \cite{2008BCJ}
\begin{gather}\label{eq:funbcj}
A_\YM(\{i,a\},n) = \sum_{a=bjc} s_{ij} A_\YM\big(ij\big(\overline b \shuffle c\big),n\big) = 0.
\end{gather}
The open string partial tree amplitudes satisfy a more complicated relation of the form \cite{2009BBDV,2010Stie}
\[ (\alpha')^{(n-3)}A_\text{string}(\{i,a\},n) = O(\alpha').\]
\end{Remark}

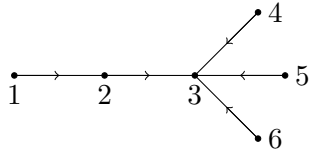
\begin{figure}[t]\centering
\begin{tikzpicture}[scale=1.2]
\draw (0,0)node[below]{1} -- (1,0)node[below]{2} -- (2,0)node[below]{3};
\draw[->] (0,0) -- (.5,0);
\draw[->] (1,0) -- (1.5,0);
\draw[->] (2.7,.7) -- (2.35,.35);
\draw[->] (2.7,-.7) -- (2.35,-.35);
\draw (2,0) -- (3,0)node[right]{5};
\draw[->] (3,0) -- (2.5,0);
\draw (2.7,.7)node[right]{4} -- (2,0) -- (2.7,-.7)node[right]{6};
\draw[fill] (0,0) circle [radius=0.03];
\draw[fill] (1,0) circle [radius=0.03];
\draw[fill] (2,0) circle [radius=0.03];
\draw[fill] (3,0) circle [radius=0.03];
\draw[fill] (2.7,.7) circle [radius=0.03];
\draw[fill] (2.7,-.7) circle [radius=0.03];
\end{tikzpicture}
\caption{The tree associated to the rational function $1/z_{12}z_{23}z_{34}z_{35}z_{36}$; and the orientation induced by designating $3$ a sink.}\label{fig:treebasis}
\end{figure}

Let $G$ be any spanning tree on the vertex set $\{1,2,\dots,n-1\}$. Designate $1$ to be the `sink' of~$G$. Then there is a unique assignment of directions to the edges of $G$, such that exactly \emph{one} edge incident on any vertex $i$ is outgoing, except for vertex $1$, which has only incoming edges (see Figure~\ref{fig:treebasis}). For a given vertex~$i$ in~$G$, let~$x(i)$ be the vertex connected to~$i$ by the one outgoing edge from~$i$.

\begin{Proposition}\label{prop:tree}
Let $G$ be as above, and fix $1$ to be the sink. Then the rational function associated to $G$ is the following product over the edges of $G$:
\begin{gather}\label{eq:TratfunG}
\prod_{\substack{{\rm edges} \\ i\rightarrow j}} \frac{1}{z_{ij}} = \prod_{i=2}^{n-2} \frac{1}{z_{ix(i)}},
\end{gather}
and this can be expanded as a sum
\begin{gather}\label{eq:Trat}
\prod_{\substack{{\rm edges} \\ i\rightarrow j}} \frac{1}{z_{ij}} = \sum_{\substack{a\\ x(i)<_{1a}i}} \pt(1a).
\end{gather}
The sum is over all orderings, $a$, such that $x(i)$ precedes $i$ in $1a$, for all $i$.
\end{Proposition}

\begin{proof}
This follows by repeated applications of Lemma \ref{lem:ptkk}. The orientation of $G$ induces a~partial order on the vertices, with~$1$ the smallest. Let~$i$ be one of the largest vertices with valence greater than~$1$, and suppose that~$i$ has~$k$ incoming edges. All vertices greater than~$i$ have valence~$1$, so that the tree `greater than~$i$' is comprised of some number of `branches', as in Figure~\ref{fig:branches}. By the lemma,
\begin{gather*}
\pt(ia)\pt(ib) = (-1)^{|b|} \pt\big(\ol{b}ia\big) = \pt(i (a\shuffle b)),
\end{gather*}
and a product of $k$ branches gives
\begin{gather*}
\pt(ia_1)\pt(ia_2)\cdots \pt(ia_k) = \pt(i (a_1\shuffle a_2\shuffle \dots \shuffle a_k)).
\end{gather*}
Moving `down the tree' gives the identity.
\end{proof}

\begin{Remark}\label{rmk:ratfuncite}
Formulas related to \eqref{eq:Trat} appear in the discussion of hyperplane arrangements in~\cite{1991SV}. The functions, \eqref{eq:TratfunG}, associated to a spanning tree $G$ are studied in \cite{2021HHTZ}, with interesting applications to CHY formulas.
\end{Remark}

\begin{figure}[t]\centering
\begin{tikzpicture}[scale=1.2]
\draw (0,0) -- (1,3) node[above]{$i$} -- (1,0) -- (1,3) -- (3,0);
\node at (1.7,1) {$\ldots$};
\draw[fill] (0,0) circle [radius=0.03];
\draw[fill] (1,3) circle [radius=0.03];
\draw[fill] (1,0) circle [radius=0.03];
\draw[fill] (3,0) circle [radius=0.03];
\draw[fill] (.5,1.5) circle [radius=0.03];
\draw[fill] (1,1) circle [radius=0.03];
\draw[fill] (1,2) circle [radius=0.03];
\end{tikzpicture}
\caption{Outermost branches of a tree.}\label{fig:branches}
\end{figure}
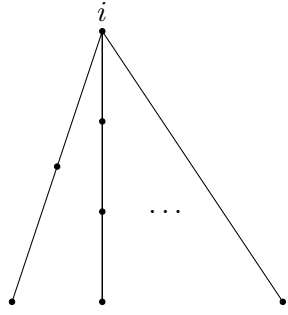

\section{Gauge and gravity tree amplitudes in BCJ form}\label{sec:4}

This section uses the identities in Section~\ref{sec:3} to obtain formulas for the tree partial amplitudes of the non-linear sigma model (NLSM), and also Yang--Mills in four dimensions. Moreover, the field theory KLT relation makes it possible to derive formulas also for the tree amplitudes of Einstein gravity, and Dirac-Born-Infeld theory, which are the gravity theories associated to YM and NLSM, respectively. There have been several previous studies that obtain formulas of this kind from CHY integrals, including \cite{2016BBBDF,2017DT,2021HHTZ}. The approach taken here is novel, and combines applications of the matrix tree theorem with the identities proved in Section~\ref{sec:3}.

\subsection{CHY formulas}\label{sec:43}
The identities in Section~\ref{sec:3} are relevant to the problem of computing tree amplitudes. This is because of the CHY formulas for the partial tree amplitudes of NLSM and YM. In~\cite{2013JulyCHY}, these formulas are written as integrals of the form
\begin{gather}\label{eq:chyintegrals}
A = \int \d\mu(a) \left(z_{ij}z_{jk}z_{ki} \prod_{\substack{l=1 \\ l\neq i,j,k}}^n \delta(E_l) \right) I,
\end{gather}
for functions $I$ of appropriate weight under the action of $\SL_2\bC$, and for any fixed choice of~$i$,~$j$,~$k$. $\SL_2\bC$ acts by M\"obius transformations on the coordinates~$z_i$. It is convenient to write
\begin{gather*}
z_{ij}z_{jk}z_{ki} \prod_{\substack{l=1 \\ l\neq i,j,k}}^n \delta(E_l) = {\prod_{l}}' \delta(E_l),
\end{gather*}
for any choice of $i$, $j$, $k$. The functions
\begin{gather}\label{eq:scatagain}
E_i = E_{i,12\dots \hat i\dots n} = \sum_{\substack{j=1\\ j\neq i}}^n \frac{s_{ij}}{z_{ij}}
\end{gather}
are the scattering equation functions introduced in Section~\ref{sec:3}. The natural logarithmic top forms on $(\bC\bP^1)^n_*$ induce volume forms, $\d\mu(a)$, on $\cM_{0,n}$: for $a=12\dots n$, the associated $\cM_{0,n}$ top form is\footnote{For any choice of a top form, $\operatorname{Vol} \PSL_2 \bC$, on the fibres of the projection $\big(\bC\bP^1\big)^n_* \rightarrow \cM_{0,n}$, such as, for example, \[ z_{ij} z_{jk}z_{ki}\, \d z_i \d z_j \d z_k,\] for any distinct~$i$,~$j$,~$k$.}
\begin{gather*}
\d\mu(123\dots n) = \frac{1}{\operatorname{Vol} \PSL_2\bC} \bigwedge_{i=1}^n \d\log z_{ii+1},
\end{gather*}
so that $\d\mu(123\dots n)$ is a $n-3$ top form on $\cM_{0,n}$ (regarded as the quotient of $\big(\bC\bP^1\big)^n_*$ by $\PSL_2 \bC$). Or, choosing the gauge fixing $z_1=0$, $z_{n-1}=1$, $z_n = \infty$,
\begin{gather*}
\d\mu(123\dots n) = \left.\frac{\d z_2 \d z_3 \cdots \d z_{n-2}}{ \prod_{i=1}^{n-2} z_{i i+1}}\right|_{z_1=0,\, z_{n-1}=1}.
\end{gather*}
The formula, \eqref{eq:chyintegrals}, must be understood as a sum of residues at the solutions of $E_i=0$. This can be written as
\begin{gather*}
\int \d\mu(a) \left({\prod_{l}}' \delta(E_l) \right) I = \sum_{\substack{\text{solutions} \\ (z)}} \operatorname{Res}_{(z)} ~ \left( \frac{\d\mu(a) I }{\det(\pa_i E_j)} \right),
\end{gather*}
where the sum is over all solutions, $(z_i^*)$, to the equations $E_i = 0$, and $\operatorname{Res}_{(z)}$ denotes the Poincar\'e residue at some given solution, $(z_i^*)$. It is left implicit in these formulas that the residue at a~solution is oriented by the form
\begin{gather*}
+\d\log E_1 \wedge \cdots \wedge \d \log E_{n-3}.
\end{gather*}

An important example is the CHY formula for the \emph{biadjoint scalar} partial tree amplitudes. In the notation of Section~\ref{sec:2}, for two orderings $a,b \in \fS(1,2,\dots,n-1)$, these partial amplitudes are
\begin{gather*}
m(a,n|b,n) = s_a (a,T(b)) = \sum_{\substack{\text{trees} \\ \alpha}} \frac{(a,\alpha)(b,\alpha)}{s_\alpha},
\end{gather*}
where the sum is over all rooted binary trees with $n-1$ labelled external edges (not including the root). The CHY formula for these amplitudes is,
\begin{gather}\label{eq:biadjoint}
m(a,n|b,n) = \int \d\mu(a)\left({\prod_{l}}' \delta(E_l) \right)\PT(b,n),
\end{gather}
where the `Parke--Taylor function' $\PT(a,n)$ is, for $a=12\dots n-1$,
\begin{gather*}
\PT(a,n) = \frac{1}{z_{12}z_{23} \cdots z_{n-1n} z_{n1}}.
\end{gather*}
Equation \eqref{eq:biadjoint} is proved in~\cite{2013DG}.

\subsection{Nonlinear sigma model}

Define an $n\times n$ matrix, $A$, with off-diagonal entires ($i\neq j$)
\begin{gather*}
A_{ij} = \frac{s_{ij}}{z_{ij}},
\end{gather*}
and diagonal entires
\begin{gather*}
A_{ii} = - \sum_{\substack{j=1 \\ j\neq i}}^n A_{ij}.
\end{gather*}
The row sums of $A$ clearly vanish. Write $A[i,j]$ for the matrix obtained by removing rows~$i$,~$j$ and columns~$i$,~$j$ from the matrix. The tree partial amplitudes for the non-linear sigma model (NLSM) are given by a CHY formula with integrand \cite{2015CHY}
\begin{gather}\label{eq:chynlsm}
A_{\text{NLSM}}(a,n) = \int \d\mu(a) \left({\prod_l}' \delta(E_l) \right) \,\frac{{\det}A[i,j]}{z_{ij}^2},
\end{gather}
for some choice of $i$, $j$.

To evaluate \eqref{eq:chynlsm}, it is convenient to choose $i=1$, $j = n$, and to gauge fix, say, $z_1=0$, $z_{n-1}=1$, $z_n = \infty$. With these choices, the integrand simplifies to
\begin{gather*}
A_{\text{NLSM}}(a,n) = \int \d^{n-3}z\,\pt(a) \left({\prod_l}' \delta(E_l) \right) \,{\det}A[1,n],
\end{gather*}
where now the diagonal entires of $A[1,n]$ are
\begin{gather*}
A_{ii} = - \sum_{\substack{j=1 \\ j\neq i}}^{n-1} A_{ij}.
\end{gather*}
To evaluate ${\det}A[1,n]$ in this limit, Kirchoff's matrix tree theorem gives~\cite{1977CK}
\begin{gather}\label{eq:Adetexpand}
\det A[1,n] = \sum_{\substack{\text{spanning}\\ \text{trees,} \\G}} \prod_{\substack{\text{edges,}\\ i\rightarrow j}} \frac{s_{ij}}{z_{ij}},
\end{gather}
where the orientations of the edges of $G$ are determined by designating vertex $1$ the sink, as in the paragraph before Proposition~\ref{prop:tree}. Given this, Proposition \ref{prop:tree} implies that
\begin{gather}\label{eq:Adettwo}
\det A[1,n] = \sum_{\substack{\text{spanning}\\ \text{trees,} \\G}} \prod_{i=2}^{n} s_{ix_i} \sum_{\substack{a\\ x_i<_{1a}i}} \pt(1a),
\end{gather}
where $x_i$ is the vertex in $G$ reached from vertex $i$ along an outgoing edge, and the second sum is over all permutations $a \in \fS(2,\dots,n-1)$ such that $x_i$ precedes $i$ in $1a$ for all $i$. After reversing the two summations,
\eqref{eq:Adettwo} becomes
\begin{gather*}
\det A[1] = \sum_{a \in \fS_{n-2}} \pt(1a) \prod_{i=2}^{n} \sum_{j<_{1a} i} s_{ij}.
\end{gather*}
The sum in the product is over all letters appearing before $i$ in the ordering~$1a$. The product appearing in this sum can be identified with components of the KLT map, using equation~\eqref{eq:Scomponent}:
\begin{gather*}
\prod_{i=2}^{n} \sum_{j<_{1a} i} s_{ix_i} = S(1a,1a).
\end{gather*}
So the NLSM partial amplitudes, \eqref{eq:chynlsm}, are
\begin{gather}\label{eq:ANLSM}
A_{\text{NLSM}} (a,n) = \sum_{b\in\fS_{n-2}} S(1b,1b) m(1b,n|a,n).
\end{gather}
As in \cite{2015CHY}, this result can be substituted into the KLT relation (equation \eqref{eq:kltrelation}) to obtain the tree amplitudes of the special Galileon theory, which is
\begin{gather*}
M_{\text{SG}} = \lim_{s_{1a}\rightarrow 0} \frac{1}{s_{1a}} \sum_{a\in\fS_{n-2}} A_{\text{NLSM}} (1a,n) S(1a,1b) A_{\text{NLSM}} (1b,n).
\end{gather*}
Using Proposition \ref{prop:TS} (or equation~\eqref{eq:TTST}), together with~\eqref{eq:ANLSM}, the formula for $M_{\text{SG}}$ can be written compactly as
\begin{gather}\label{eq:DBIKLT2}
M_{\text{SG}} = \sum_{a\in\fS_{n-2}} S(1a,1a) A_{\text{NLSM}} (1a,n).
\end{gather}

\begin{Proposition}\label{prop:nlsm}
The NLSM and special Galileon tree amplitudes can be expressed as the following sums over binary trees:
\begin{gather*}
A_{{\rm NLSM}} (a,n) = \sum_{\substack{{\rm trees,}\\ \alpha}} \frac{ (a,\alpha) n_\alpha}{s_\alpha}\qquad \text{and}\qquad M_{{\rm SG}} = \sum_{\substack{{\rm trees,}\\ \alpha}} \frac{ n_\alpha n_\alpha}{s_\alpha},
\end{gather*}
where
\begin{gather}\label{eq:nalphaS}
n_\alpha = \sum_{b\in \fS_{n-2}} (1b,\alpha) S(1b,1b).
\end{gather}
The numerators $n_\alpha$ have no poles in the~$s_I$ variables, and the replacement $\alpha \mapsto n_\alpha$ defines a~homomorphism out of~$\cL_A$. Numerators \end{Proposition}

The formula \eqref{eq:nalphaS} was found also in~\cite{2017CMS,2020Mafra}; the derivation here is new. One interest of this derivation is that the methods used here can also be easily adapted to study other CHY integrals. The next section discusses the case of Yang--Mills gauge theory in four dimensions.

\subsection{Yang--Mills}

Some of the methods used in the previous section to study NLSM amplitudes also lead to results about Yang--Mills tree amplitudes. This is because there exist formulas for Yang--Mills tree amplitudes (in four dimensions) that involve determinants similar to those computed above. This section first recalls these formulas, and then manipulates them using the identities from Section~\ref{sec:3}.

In four dimensions, gluons have two helicity states, and it is conventional to further refine the partial amplitude decomposition by specifying the helicities of the gluons. Using spinors, the null momenta $k_i$ may be written as $k_i^{\alpha\dot\alpha} = \lambda_i^\alpha \tilde\lambda_i^{\dot\alpha}$, unique up to a complex rescaling $\lambda_i \rightarrow \alpha \lambda_i$, $\tilde\lambda_i \rightarrow \alpha^{-1} \tilde\lambda_i$ (with $\alpha \neq 0$). The two helicities correspond two polarizations $\epsilon_+^{\alpha\dot\alpha} \propto \lambda^\alpha \tilde\xi^{\dot\alpha}$ and $\epsilon_-^{\alpha\dot\alpha} \propto\xi^{\alpha} \tilde\lambda^{\dot\alpha}$, for some reference spinors $\xi$ and $\dot\xi$. In practice, the partial amplitudes themselves do not depend on the choice of reference spinors, and are functions of the invariants
\begin{gather*}
\avg{ij} = \lambda_i^\alpha \lambda_j^\beta \epsilon_{\alpha\beta},\qquad [ij] = \tilde\lambda_i^{\dot\alpha} \tilde\lambda_j^{\dot\beta} \epsilon_{\dot\alpha\dot\beta},
\end{gather*}
where $\epsilon_{12} = - \epsilon_{21} = 1$. Fix $k$ gluons `$1,\dots,k$' with $+$ helicity, and $n-k$ gluons `$k+1,\dots,n$' with $-$ helicity.

The CHY-like formulas for 4D amplitudes that we will consider are sums over solutions to so-called `polarized scattering equations', which have fewer solutions than the scattering equations. To write these equations, it is helpful to define the following two spinor-valued functions on~$\bC\bP^1$:
\begin{gather*}
\lambda^\alpha(z) = \sum_{j=k+1}^n \frac{t_j \lambda_j^\alpha}{z-z_j}\qquad\text{and}\qquad \tilde\lambda^{\dot\alpha}(z) = \sum_{i=1}^k \frac{t_i \tilde \lambda^{\dot\alpha}_i}{z - z_i},
\end{gather*}
where the $t_i$ are non-zero complex scalars. Then the polarized scattering equations are
\begin{gather}
\lambda_i^\alpha - t_i \lambda^\alpha(z_i) = 0, \qquad \text{for } i = 1,\dots,k, \nonumber\\
\tilde\lambda^{\dot\alpha}_j - t_j \tilde\lambda^{\dot\alpha}(z_j) = 0, \qquad \text{for } j = k+1, \dots, n.\label{eq:polscat2}
\end{gather}
Solutions to these equations are also solutions to the `original scattering equations', \eqref{eq:scatagain}, because \eqref{eq:polscat2} ensure that the following residues vanish:
\begin{gather*}
\Res_{z=z_i}\, k_i^{\alpha\dot\alpha} \lambda_\alpha(z) \tilde\lambda_{\dot\alpha}(z).
\end{gather*}
The equations \eqref{eq:polscat2} also imply that the spinor data satisfies momentum conservation,
\begin{gather*}
\sum_{i=1}^n \lambda_i^\alpha \tilde\lambda_i^{\dot\alpha} = 0,
\end{gather*}
which can be checked by breaking the sum into two parts: $i=1,\dots,k$ and $i = k+1,\dots,n$.

In order to present the formula for 4D gravity amplitudes, first define the $k\times k$ Hodges' matrix, $H$: the off-diagonal entries are~\cite{2013CS,2012H}
\begin{gather*}
H_{ij} = \frac{ t_it_j \avg{ij}}{z_{ij}} \qquad \text{for $i\neq j$},
\end{gather*}
and the diagonal entries are
\begin{gather*}
H_{ii} = - \sum_{\substack{i=1\\j\neq i}}^{k} H_{ij}.
\end{gather*}
Likewise, define the $(n-k)\times (n-k)$ matrix $\tilde H$ to have entries
\begin{gather*}
\tilde{H}_{ij} = \frac{\tilde t_i \tilde t_j [ij]}{z_{ij}} \qquad \text{for $i\neq j$},
\end{gather*}
and
\begin{gather*}
\tilde{H}_{ii} = - \sum_{\substack{j=n-k\\ j\neq i}}^{n} \tilde{H}_{ij}.
\end{gather*}

Fix some $a$ in $1,\dots,k$, and fix some $b$ in $k+1,\dots,n$. Write $H[a]$ for the matrix formed by removing the $a^{\text{th}}$ row and column from $H$. Likewise for $\tilde H[b]$. Given these Hodges' matrices, the 4D gravity tree amplitude can be written as
\begin{gather}\label{eq:specGR}
M_{\text{GR}} = \int \d\mu\, \frac{{\det} H[a] {\det} \tilde H[b]}{\prod_{i=1}^n t_i^2} \prod_{\substack{i=1\\i\neq a}}^k \delta^2(\lambda_i - t_i \lambda(z_i)) \prod_{\substack{j=k+1\\j\neq b}}^n \delta^2\big(\tilde \lambda_i - \tilde t_i \tilde \lambda(z_i)\big).
\end{gather}
This formula as given can be found in \cite{2015ACRS,2014GLM}, and is proved in \cite{2016Geyer}. It is equivalent to and closely related to the Cachazo--Skinner formula \cite{2014CMS,2013CS}, and also to the RSVW formula~\cite{2004RSV}. The 4D Yang--Mills tree partial amplitude can likewise be written as
\begin{gather}\label{eq:specym}
A_{\text{YM}}(an) = \int \d\mu\, \PT(an) \prod_{\substack{i=1\\i\neq a}}^k \delta^2(\lambda_i - t_i \lambda(z_i)) \prod_{\substack{j=k+1\\j\neq b}}^n \delta^2\big(\tilde \lambda_i - \tilde t_i \tilde \lambda(z_i)\big).
\end{gather}
$\GL_2\bC$ acts on the integrands in~\eqref{eq:specGR} and~\eqref{eq:specym} by matrix multiplication on the pair $(z_i/t_i, 1/t_i)$. In both formulas, the measure~$\d\mu$ is given by
\begin{gather*}
\d\mu = \frac{1}{\text{Vol} \GL_2\bC} \prod_{i=1}^n \frac{\d z_i \d t_i}{t_i},
\end{gather*}
so that $\d\mu$ is a top dimensional $2n-4$ form. More explicitly, adopting the gauge fixing $z_1=0$, $z_2=1$, $z_n = \infty$, $t_n = 1$,
\begin{gather*}
\d\mu = \left.\prod_{i=3}^{n-1} \frac{\d z_i \d t_i}{t_i}\right|_{z_1=0,\, z_2=1, \,z_n = \infty,\, t_n = 1}.
\end{gather*}

\begin{Remark}\label{rmk:glcov}
By defining homogeneous coordinates $\sigma_i = (z_i/t_i, 1/t_i)$, it is possible to write the integrands of \eqref{eq:specGR} and \eqref{eq:specym} in $\GL_2\bC$-covariant form, using the pairing
\begin{gather*}
(\sigma_i,\sigma_j) = \frac{z_i - z_j}{t_i t_j}.
\end{gather*}
These leads to the formulas in the form presented in \cite{2015ACRS,2014GLM}.
\end{Remark}

The rest of this section uses the tools from Section~\ref{sec:3} to expand the integrand in \eqref{eq:specGR}, in order to write $M_{\text{GR}}$ as a sum
\begin{gather}\label{eq:expandnym}
M_{\text{GR}} = \sum_{a\in \fS_{n-2}} N_{\text{YM}}(1an) A_\text{YM}(1an).
\end{gather}
Computing the coefficients of this expansion, $N_{\text{YM}}(1an)$, also suffices to compute $A_\text{YM}$ itself: this follows from the KLT relation, \eqref{eq:kltrelation}, as seen in the case of NLSM amplitudes, in equations~\eqref{eq:ANLSM} to~\eqref{eq:DBIKLT2}.

To compute $N_{\text{YM}}(1an)$, the first step is to expand the determinants $\det H[a]$ and $\det\tilde H[b]$. The can be done using Kirkchoff's tree theorem, as in~\eqref{eq:Adetexpand}, above, with the difference that $H$ is a \emph{symmetric} matrix, whereas~$A$ is not. Fix some $a$ from $1,\dots,k$. Then the determinant of~$H[a]$ is
\begin{gather}\label{eq:Hdetexpand}
\det H[a] = \sum_{\substack{\text{trees} \\ G}} \prod_{\substack{\text{edges}\\i-j}} \frac{t_i t_j \avg{ij}}{z_{ij}},
\end{gather}
where the sum is over all spanning trees, $G$, of the vertex set $1,\dots,k$. The Hodges matrix $H$ is symmetric, so it is not necessary to orient the edges in order for \eqref{eq:Hdetexpand} to be well defined. Also note that the result, \eqref{eq:Hdetexpand}, is independent of the choice of~$a$. A single summand in~\eqref{eq:Hdetexpand}, for a tree $G$, can also be written as
\begin{gather*}
\left( \prod_{i=1}^k t_i^{d_i} \right) \prod_{\substack{\text{edges}\\i-j}} \frac{\avg{ij}}{z_{ij}},
\end{gather*}
where $d_i$ is the degree of the vertex $i$ in $G$. The determinant $\det' \tilde H$ is a sum of similar such terms. Fixing some $b$ from $k+1,\dots,n$,{\samepage
\begin{gather*}
\det \tilde H[b] = \sum_{\substack{\text{trees} \\ G}} \, \left( \prod_{j=k+1}^n t_j^{d_j} \right) \prod_{\substack{\text{edges}\\i-j \\ \text{in $G$}}} \frac{[ij]}{z_{ij}},
\end{gather*}
where $d_j$ is the degree of the vertex $j$ in the spanning tree $G$.}

The full amplitude, $M_\text{GR}$, can therefore be expanded as a sum over pairs of trees $(G,G')$, with $G$ spanning vertex set $1,\dots,k$ and $G'$ spanning vertex set $k+1,\dots,n$. Explicitly,
\begin{gather}\label{eq:MGRfullexpand}
M_\text{GR} = \sum_{\substack{\text{trees}\\G,G'}} \int \d\mu\, I_{G,G'} \prod_{\substack{i=1\\i\neq a}}^k \delta^2(\lambda_i - t_i \lambda(z_i)) \prod_{\substack{j=k+1\\j\neq b}}^n \delta^2\big(\tilde \lambda_i - \tilde t_i \tilde \lambda(z_i)\big),
\end{gather}
where
\begin{gather*}
I_{G,G'} = \left( \prod_{i=1}^n t_i^{d_i-2} \right) \left(\prod_{\substack{i-j \\ \text{in $G$}}} \frac{\avg{ij}}{z_{ij}} \right) \left(\prod_{\substack{i-j \\ \text{in $G'$}}} \frac{[ij]}{z_{ij}} \right).
\end{gather*}
The \looseness=-1 coefficients $n_{\text{YM}}(1a)$ in \eqref{eq:expandnym} can in principle be computed from \eqref{eq:MGRfullexpand} in two steps. First, it is necessary to express the $t_i$ in terms of the $z_i$ by solving for them using the polarized scattering equations. Second, the integrands $I_{G,G'}$ should be expanded in Parke--Taylor factors using the identities in Section \ref{sec:3}. The resulting terms can then be regrouped to give a sum of the form \eqref{eq:expandnym}. This is carried out for the maximal-helicity-violating (MHV) case in the following subsection.

\subsection{Maximal-helicity-violating Yang--Mills amplitude}

The Maximal-helicity-violating (MHV) case is $k=2$, when only two gluons are $+$ helicity, and $n-2$ are $-$ helicity. In this case, the matrix tree expansion, \eqref{eq:MGRfullexpand}, simplifies to the following sum over spanning trees on $3,\dots,n$:
\begin{gather*}
M_\text{GR} = \sum_{\substack{\text{trees}\\G}} \int \d\mu\, I_{G} \prod_{\substack{i=1\\i\neq a}}^k \delta^2(\lambda_i - t_i \lambda(z_i)) \prod_{\substack{j=k+1\\j\neq b}}^n \delta^2\big(\tilde \lambda_i - \tilde t_i \tilde \lambda(z_i)\big),
\end{gather*}
where
\begin{gather*}
I_{G} = \frac{\avg{12}}{z_{12}} \left(\frac{1}{t_1t_2} \prod_{i=3}^n t_i^{d_i-2} \right) \left(\prod_{\substack{i-j \\ \text{in $G$}}} \frac{[ij]}{z_{ij}} \right).
\end{gather*}
A tree on $n-2$ vertices has $n-3$ edges, so that
\begin{gather*}
\sum_{i=3}^n (d_i -2) = 2(n-3) - 2(n-2) = -2.
\end{gather*}
It follows that $I_G$ may also be written as
\begin{gather*}
I_{G} = \frac{\avg{12}}{z_{12}} \left(\frac{t_1}{t_2} \prod_{i=3}^n (t_1t_i)^{d_i-2} \right) \left(\prod_{\substack{i-j \\ \text{in $G$}}} \frac{[ij]}{z_{ij}} \right).
\end{gather*}
Now fix some $b$ from $3,\dots,n$. Choosing $b$ to be a source vertex induces an orientation of each spanning tree $G$, such that every vertex (apart from~$b$) has exactly~$1$ incoming edge. Given that~$G$ is oriented this way, $I_G$ may be further re-written as
\begin{gather}\label{mhv:IG3}
I_{G} = \frac{\avg{12}}{z_{12}} \left(\frac{t_1}{t_2} (t_1t_b)^{-2} \prod_{\substack{i=3\\ i\neq b}}^n (t_1t_i)^{-1} \right) \left(\prod_{\substack{i\rightarrow j \\ \text{in $G$}}} \frac{(t_1t_i) [ij]}{z_{ij}} \right).
\end{gather}

Having used the matrix tree theorem to evaluate $I_G$, the next step is to use the polarized scattering equations to solve for the $t_i$ appearing in~\eqref{mhv:IG3}. In the MHV case, the polarized scattering equations include,
\begin{gather*}
\tilde\lambda_j = t_j \left( \frac{t_1 \tilde\lambda_1}{z_{i1}} + \frac{t_2 \tilde\lambda_2}{z_{i2}} \right),
\end{gather*}
for $j$ in $3,\dots,n$. These equations imply that
\begin{gather*}
t_1t_j = \frac{[j2]}{[12]} z_{j1},
\end{gather*}
and that, for any given $j$ in $3,\dots,n$,
\begin{gather*}
\frac{t_1}{t_2} = - \frac{[j2]}{[j1]} \frac{z_{j1}}{z_{j2}}.
\end{gather*}
Making these substitutions, it follows that, on the support of the polarized scattering equations,
\begin{gather*}
I_G = - \frac{\avg{12}}{z_{12}} \left( \frac{[b2]}{[b1]} \frac{z_{b1}}{z_{b2}}\right) \left(\frac{[12]}{[b2]} \frac{1}{z_{b1}}\right)^{2} \left( \prod_{\substack{j=3\\ j\neq b}}^n \frac{[12]}{[j2]} \frac{1}{z_{j1}} \right) \left(\prod_{\substack{i\rightarrow j \\ \text{in $G$}}} \frac{[i2][ij]}{[12]} \frac{z_{i1}}{z_{ij}} \right).
\end{gather*}
Both of the products appearing in this equation have $n-3$ terms, and the factors of~$[12]$ in them cancel out. Combining the remaining factors gives
\begin{gather}\label{eq:newIGG}
I_G = \frac{\avg{12}[12]^2}{[b1][b2]} \frac{1}{z_{12}z_{2b}z_{b1}} \prod_{\substack{i\rightarrow j \\ \text{in $G$}}} \left( \frac{[i2][ij]}{[j2]} \frac{z_{i1}}{z_{j1} z_{ij}} \right).
\end{gather}

Proposition \ref{prop:tree} can be used to reexpress the product in \eqref{eq:newIGG} in terms of Parke--Taylor functions. A variation on the argument in Proposition \ref{prop:tree} gives\footnote{This follows by telescoping the factors of $z_{i1}/z_{j1}$.}
\begin{gather*}
\prod_{\substack{i\rightarrow j \\ \text{in $G$}}} \frac{z_{i1}}{z_{j1} z_{ij}} = \sum_{\substack{\sigma \\ x_i < i}} \pt(b\sigma) \frac{z_{b1}}{z_{\sigma^*1}},
\end{gather*}
where $\sigma^*$ is the last entry of $\sigma$. The resulting expression for $I_G$ is
\begin{gather*}
I_G = \frac{\avg{12}[12]^2}{[b1][b2]}\left( \prod_{\substack{i\rightarrow j \\ \text{in $G$}}} \, \frac{[i2][ij]}{[j2]} \right) \sum_{\substack{\sigma \\ x_i < i}} \PT(12b\sigma).
\end{gather*}
Or, reordering the summations, it follows that, on the support of the polarized scattering equations,
\begin{gather*}
\sum_G I_G = \sum_{\sigma \in \fS_{n-3}} \PT(12b\sigma) N(12b\sigma),
\end{gather*}
where
\begin{gather}\label{mhv:NNN}
N(12b\sigma) = \frac{\avg{12}[12]^2}{[b1][b2]} \prod_{\substack{j=3 \\ i\neq b}}^n \sum_{i<_\sigma j} \frac{[i2][ij]}{[j2]}.
\end{gather}
This gives an expansion of the gravity tree amplitude into a sum of YM partial amplitudes,
\begin{gather*}
M_\text{GR} = \sum_{\sigma \in \fS_{n-3}} A_\text{YM}(12b\sigma) N(12b\sigma).
\end{gather*}
For example, when combined with the Parke--Taylor formula for the MHV Yang--Mills amplitude, $A_{\text{YM}}(12b\sigma)$, this yields, via~\eqref{eq:specym} and~\eqref{eq:specGR}, the MHV gravity amplitude
\begin{gather*}
M_{\text{GR}}(1234) = \frac{\avg{12}[12]^8}{[34]}\left( \prod_{\substack{i<j \\ i=1}}^4 [ij] \right)^{-1}.
\end{gather*}

\begin{Remark}\label{rmk:bgk}
Formulas for the MHV gravity amplitude were first obtained by Berends--Giele--Kuijf (BGK), and were inspired by the KLT relations~\cite{1988BGK}. Variations on the BGK formula have appeared in several places, including \cite{2005BBST,2008EF,2009MS}. However, the formula above appears to be a new variant, and it has been derived here using new methods motivated by Section~\ref{sec:3}. The studies cited above mostly use inductive arguments based on recursion relations.
\end{Remark}

\section{Discussion}\label{sec:5}
The formulas for partial NLSM tree amplitudes in Section \ref{sec:4} can be expressed in the form
\begin{gather}\label{eq:generalbcj}
A(a,n) = \sum_{\text{trees}~\alpha} \frac{(a,\alpha)n_\alpha}{s_\alpha},
\end{gather}
for numerators $n_\alpha$ that satisfy two key properties. First, the replacement $\alpha \mapsto n_\alpha$ extends to define a homomorphism out of the space of Lie polynomials, $\cL_A$. Second, the $n_\alpha$ are \emph{polynomial} in the Mandelstam variables~$s_I$. The numerators $n_\alpha$, satisfying these two properties, are called `BCJ numerators', after~\cite{2008BCJ}. The results for MHV gravity also lead to formulas of the form~\eqref{eq:generalbcj} for YM MHV amplitudes, but the $n_\alpha$ have spurious poles coming from the denominator factors in equation~\eqref{mhv:NNN}. Beyond giving formulas for these numerators, there are two important unresolved questions about the numerators for further research.

First, the BCJ numerators in Section~\ref{sec:4} were of the form
\begin{gather*}
n_\alpha = \sum_a (1a,\alpha) n(1a),
\end{gather*}
for some functions $n(1a)$. A number of authors have asked whether there exists a `kinematic' Lie algebra such that the $n_\alpha$ can be expressed instead as a Lie bracketing of $n-1$ Lie algebra elements, by analogy with the definition of~$c_\alpha$,~\eqref{eq:colourfactordef}, \cite{2019MB,2021CJTW,2017CS,2014MC}. The formulas obtained using the methods in this paper may suggest further clues for identifying such `kinematic algebras' for gauge theories like Yang--Mills and NLSM.

\looseness=-1 Second, as has been widely observed, if the partial amplitudes of a gauge theory can be written in the form~\eqref{eq:generalbcj}, for BCJ numerators $n_\alpha$, then that gauge theory can participate in a~KLT relation to produce a gravity amplitude (for some gravity-like theory). It would therefore be desirable to characterise or classify all gauge theories whose tree amplitudes can be obtained using BCJ numerators. A full answer to this question should consider the space of all possible perturbative gauge theories, which is beyond the immediate scope of the methods used in this paper.

These questions are all concerned with the KLT relations satisfied by \emph{tree level} amplitudes. An important further aim is to discover whether the KLT relation,~\eqref{eq:kltrelation}, can be extended to a~statement about higher order terms in the perturbation series. There have been some attempts to formulate a KLT relation amoung 1-loop amplitudes, both in string theory and gauge theory, and the Lie bracket $\{~,~\}$ studied in Section~\ref{sec:2} plays some role here (discussed in~\cite{2020MS}). The key, however, to understanding the tree-level KLT relation is the properties of the colour factors and the partial amplitude decomposition, which at tree level is easily understood as arising from Lie polynomials. At higher orders in perturbation series, the partial amplitude decomposition is related to the topology of surfaces. How this arises is reviewed in Appendix~\ref{app:A}. To formulate KLT relations at higher loop order, it will be useful to understand the algebraic properties of the colour factors and partial amplitude decomposition. By analogy with the tree level case, it would be very difficult to arrive at the tree level KLT relations without understanding the Kleiss--Kuijf,~\eqref{eq:kksha}, and `fundamental BCJ' relations,~\eqref{eq:funbcj}, among partial amplitudes. Finding the analogs of these relations at higher orders in perturbation theory would therefore be a good starting point for further work on this topic.

\appendix
\section{Colour factors and the partial amplitude decomposition}\label{app:A}
Fix a Feynman diagram, $D$, for some ${\rm SU}(N)$ gauge theory. Let $D$ has $k$ internal vertices and $n$ labelled external lines and, for simplicity, suppose $D$ has only cubic vertices. Let $F_1,\dots,F_n \in \operatorname{ad}(\fsu(N))$ be the colour states (in the adjoint representation) associated to each external line. The contribution of $D$ to the amplitude factors as
\begin{gather*}
A_D = (ig)^k C_D I_D,
\end{gather*}
where $C_D$ is some invariant function of the $F_i$ and $I_D$ is the Feynman integral associated to the graph. As a tensor diagram, the commutator of two $F_i$'s can be written as
\begin{equation*}
\begin{tikzpicture}
\draw node at (-1,0.75){$[F_1,F_2] =$};
\draw (0,0) to [out=90, in=-90] (0.9,1.5);
\draw (0.2,0) to [out=90, in=180] (1,0.7) to [out=0,in=90] (1.8,0);
\draw (2,0) to [out=90, in=-90] (1.1,1.5);
\draw node at (0.1,-0.3) {$F_1$};
\draw node at (1.9,-0.3) {$F_2$};
\draw node at (2.5,0.75){$-$};
\draw (3,0) to [out=90, in=-90] (3.9,1.5);
\draw (3.2,0) to [out=90, in=180] (4,0.7) to [out=0,in=90] (4.8,0);
\draw (5,0) to [out=90, in=-90] (4.1,1.5);
\draw node at (3.1,-0.3) {$F_2$};
\draw node at (4.9,-0.3) {$F_1$};
\end{tikzpicture}
\end{equation*}
It follows that $C_D$ can be expanded as a sum of $2^{k-1}$ terms, corresponding to $2^{k-1}$ choices of cyclic orientation to assign to each vertex of $D$. Each term corresponds to a cubic ribbon graph,~$G$, that retracts onto~$D$. Write $c_G$ for the contraction of the $F_i$ according to the ribbon graph~$G$, regarded as a tensor diagram. Then
\begin{gather*}
C_D = \sum_{\substack{G \text{ s.t.}\\ \Thin_G = D} } (-1)^{|G|} c_G,
\end{gather*}
where $|G|$ is the number of white vertices of $G$, with the vertex connected to $n$ fixed to be black, say. In fact, $c_G$ depends only on $G$ regarded as a topological surface (forgetting the graph). Each closed boundary component of $G$ is a trace. A boundary component with marked points $1,\dots,k$ arranged cyclically, gives a contribution
\begin{gather*}
\tr \big(F^1F^2\dots F^k\big).
\end{gather*}
Whereas a boundary component of $G$ with no marked points gives a trace of the identity, which is $\tr (\text{Id}) = N$. Since $c_G$ does not depend on the graph structure of $G$, write $c_\Sigma$ for the colour factor associated to a marked surface with boundary, $\Sigma$. Collecting terms, it is possible to write the whole perturbation series as a sum over surfaces (just as happens for the open string amplitude),
\begin{gather*}
A(1,\dots,n) = \sum A_\Sigma c_\Sigma,
\end{gather*}
where the surfaces $\Sigma$ have boundary marked points labelled by $1,\dots,n$. Likewise, there is a~partial amplitude series for biadjoint scalar theory,
\begin{gather*}
A_{\phi^3}(1,\dots,n) = \sum_D A_D = \sum_{\Sigma, \Sigma'} c_\Sigma \tilde c_{\Sigma'} A(\Sigma,\Sigma').
\end{gather*}
The double partial amplitudes $A(\Sigma,\Sigma')$ are the analog, at higher orders in perturbation theory, of the matrix $m(1a,n|1b,n) = s_{1a}(1a,T(1b))$ whose inverse (away from $s_{1a}=0$) produced the field theory KLT kernel in Section~\ref{sec:2}. It is reasonable to expect that there exists a KLT relation for the partial amplitudes $A(\Sigma)$ at fixed loop order, involving genus $g$ surfaces, with $h$ boundaries and $p$ punctures, subject to the Euler characteristic constraint,
\begin{gather*}
p+2g+h = \ell+1.
\end{gather*}

\subsection*{Acknowledgements}
This article is based on a talk given at the 2020 `Algebraic Structures in Perturbative Quantum Field Theory' conference, in honor of Dirk Kreimer for his 60th birthday. The author thanks Lionel Mason and the referees for their crucial observations and helpful comments. The work is supported by ERC grant GALOP ID: 724638.

\pdfbookmark[1]{References}{ref}
\LastPageEnding


\begin{thebibliography}{99}
\footnotesize\itemsep=0pt

\bibitem{2015ACRS}
Adamo T., Casali E., Roehrig K.A., Skinner D., On tree amplitudes of
 supersymmetric {E}instein--{Y}ang--{M}ills theory, \href{https://doi.org/10.1007/jhep12(2015)177}{\textit{J.~High Energy
 Phys.}} \textbf{2015} (2015), no.~12, 177, 14~pages, \href{https://arxiv.org/abs/1507.02207}{arXiv:1507.02207}.

\bibitem{2005BBST}
Bedford J., Brandhuber A., Spence B., Travaglini G., A recursion relation for
 gravity amplitudes, \href{https://doi.org/10.1016/j.nuclphysb.2005.05.016}{\textit{Nuclear Phys.~B}} \textbf{721} (2005), 98--110,
 \href{https://arxiv.org/abs/hep-th/0502146}{arXiv:hep-th/0502146}.

\bibitem{1988BGK}
Berends F.A., Giele W.T., Recursive calculations for processes with $n$ gluons,
 \href{https://doi.org/10.1016/0550-3213(88)90442-7}{\textit{Nuclear Phys.~B}} \textbf{306} (1988), 759--808.

\bibitem{2008BCJ}
Bern Z., Carrasco J.J.M., Johansson H., New relations for gauge-theory
 amplitudes, \href{https://doi.org/10.1103/PhysRevD.78.085011}{\textit{Phys. Rev.~D}} \textbf{78} (2008), 085011, 19~pages,
 \href{https://arxiv.org/abs/0805.3993}{arXiv:0805.3993}.

\bibitem{2016BBBDF}
Bjerrum-Bohr N.E.J., Bourjaily J.L., Damgaard P.H., Feng B., Manifesting
 color-kinematics duality in the scattering equation formalism,
 \href{https://doi.org/10.1007/JHEP09(2016)094}{\textit{J.~High Energy Phys.}} \textbf{2016} (2016), no.~9, 094, 17~pages,
 \href{https://arxiv.org/abs/1608.00006}{arXiv:1608.00006}.

\bibitem{2010BBDSV}
Bjerrum-Bohr N.E.J., Damgaard P.H., S{\o}ndergaard T., Vanhove P., The momentum
 kernel of gauge and gravity theories, \href{https://doi.org/10.1007/JHEP01(2011)001}{\textit{J.~High Energy Phys.}}
 \textbf{2011} (2011), no.~1, 001, 18~pages, \href{https://arxiv.org/abs/1010.3933}{arXiv:1010.3933}.

\bibitem{2009BBDV}
Bjerrum-Bohr N.E.J., Damgaard P.H., Vanhove P., Minimal basis for gauge theory
 amplitudes, \href{https://doi.org/10.1103/PhysRevLett.103.161602}{\textit{Phys. Rev. Lett.}} \textbf{103} (2009), 161602, 4~pages,
 \href{https://arxiv.org/abs/0907.1425}{arXiv:0907.1425}.

\bibitem{2019MB}
Bridges E., Mafra C.R., Algorithmic construction of {SYM} multiparticle
 superfields in the {BCJ} gauge, \href{https://doi.org/10.1007/jhep10(2019)022}{\textit{J.~High Energy Phys.}} \textbf{2019}
 (2019), no.~10, 022, 35~pages, \href{https://arxiv.org/abs/1906.12252}{arXiv:1906.12252}.

\bibitem{1999BV}
Brion M., Vergne M., Arrangement of hyperplanes. {I}.~{R}ational functions and
 {J}effrey--{K}irwan residue, \href{https://doi.org/10.1016/S0012-9593(01)80005-7}{\textit{Ann. Sci. \'Ecole Norm. Sup.~(4)}}
 \textbf{32} (1999), 715--741, \href{https://arxiv.org/abs/math.DG/9903178}{arXiv:math.DG/9903178}.

\bibitem{2018BD}
Brown F., Dupont C., Single-valued integration and superstring amplitudes in
 genus zero, \href{https://doi.org/10.1007/s00220-021-03969-4}{\textit{Comm. Math. Phys.}} \textbf{382} (2021), 815--874,
 \href{https://arxiv.org/abs/1910.01107}{arXiv:1910.01107}.

\bibitem{2012C}
Cachazo F., Fundamental {BCJ} relation in {${\mathcal N}=4$} {SYM} from the
 connected formulation, \href{https://arxiv.org/abs/1206.5970}{arXiv:1206.5970}.

\bibitem{2013JulyCHY}
Cachazo F., He S., Yuan E.Y., Scattering of massless Particles in arbitrary
 dimension, \href{https://doi.org/10.1103/PhysRevLett.113.171601}{\textit{Phys. Rev. Lett.}} \textbf{113} (2014), 171601, 4~pages,
 \href{https://arxiv.org/abs/1307.2199}{arXiv:1307.2199}.

\bibitem{2015CHY}
Cachazo F., He S., Yuan E.Y., Scattering equations and matrices: from
 {E}instein to {Y}ang--{M}ills, {DBI} and {NLSM}, \href{https://doi.org/10.1007/JHEP07(2015)149}{\textit{J.~High Energy
 Phys.}} \textbf{2015} (2015), no.~7, 149, 43~pages, \href{https://arxiv.org/abs/1412.3479}{arXiv:1412.3479}.

\bibitem{2014CMS}
Cachazo F., Mason L., Skinner D., Gravity in twistor space and its
 {G}rassmannian formulation, \href{https://doi.org/10.3842/SIGMA.2014.051}{\textit{SIGMA}} \textbf{10} (2014), 051, 28~pages,
 \href{https://arxiv.org/abs/1207.4712}{arXiv:1207.4712}.

\bibitem{2013CS}
Cachazo F., Skinner D., Gravity from rational curves in twistor space,
 \href{https://doi.org/10.1103/PhysRevLett.110.161301}{\textit{Phys. Rev. Lett.}} \textbf{110} (2013), 161301, 4~pages,
 \href{https://arxiv.org/abs/1207.0741}{arXiv:1207.0741}.

\bibitem{2017CMS}
Carrasco J.J.M., Mafra C.R., Schlotterer O., Abelian {$Z$}-theory: {NLSM}
 amplitudes and {$\alpha'$}-corrections from the open string, \href{https://doi.org/10.1007/JHEP06(2017)093}{\textit{J.~High
 Energy Phys.}} \textbf{2017} (2017), no.~6, 093, 31~pages,
 \href{https://arxiv.org/abs/1608.02569}{arXiv:1608.02569}.

\bibitem{1977CK}
Chaiken S., Kleitman D.J., Matrix tree theorems, \href{https://doi.org/10.1016/0097-3165(78)90067-5}{\textit{J.~Combinatorial
 Theory Ser.~A}} \textbf{24} (1978), 377--381.

\bibitem{2021CJTW}
Chen G., Johansson H., Teng F., Wang T., Next-to-{MHV} {Y}ang--{M}ills
 kinematic algebra, \href{https://doi.org/10.1007/JHEP10(2021)042}{\textit{J.~High Energy Phys.}} \textbf{2021} (2021),
 no.~10, 042, 54~pages, \href{https://arxiv.org/abs/2104.12726}{arXiv:2104.12726}.

\bibitem{2017CS}
Cheung C., Shen C.-H., Symmetry for flavor-kinematics duality from an action,
 \href{https://doi.org/10.1103/PhysRevLett.118.121601}{\textit{Phys. Rev. Lett.}} \textbf{118} (2017), 121601, 5~pages,
 \href{https://arxiv.org/abs/1612.00868}{arXiv:1612.00868}.

\bibitem{2013DG}
Dolan L., Goddard P., Proof of the formula of {C}achazo, {H}e and {Y}uan for
 {Y}ang--{M}ills tree amplitudes in arbitrary dimension, \href{https://doi.org/10.1007/JHEP05(2014)010}{\textit{J.~High
 Energy Phys.}} \textbf{2014} (2014), no.~5, 010, 24~pages, \href{https://arxiv.org/abs/1311.5200}{arXiv:1311.5200}.

\bibitem{2017DT}
Du Y.-J., Teng F., B{CJ} numerators from reduced {P}faffian, \href{https://doi.org/10.1007/JHEP04(2017)033}{\textit{J.~High
 Energy Phys.}} \textbf{2017} (2017), no.~4, 033, 28~pages,
 \href{https://arxiv.org/abs/1703.05717}{arXiv:1703.05717}.

\bibitem{2008EF}
Elvang H., Freedman D.Z., Note on graviton {MHV} amplitudes, \href{https://doi.org/10.1088/1126-6708/2008/05/096}{\textit{J.~High
 Energy Phys.}} \textbf{2008} (2008), no.~5, 096, 11~pages, \href{https://arxiv.org/abs/0710.1270}{arXiv:0710.1270}.

\bibitem{Thesis}
Frost H., Universal aspects of perturbative gauge theory amplitudes, Ph.D.~Thesis, {U}niversity of Oxford, 2020, available at
 \url{https://ora.ox.ac.uk/objects/uuid:2b7fc6f9-ee97-40b0-9c14-1f29dfb916fc}.

\bibitem{2020FrostMM}
Frost H., Mafra C.R., Mason L., A {L}ie bracket for the momentum kernel,
 \href{https://arxiv.org/abs/2012.00519}{arXiv:2012.00519}.

\bibitem{2016Geyer}
Geyer Y., Ambitwistor strings: worldsheet approaches to perturbative quantum
 field theories, Ph.D.~Thesis, {U}niversity of Oxford, 2016, available at
 \url{https://ora.ox.ac.uk/objects/uuid:ff543496-58bd-4818-ae0c-cf31eb349c90}.

\bibitem{2014GLM}
Geyer Y., Lipstein A.E., Mason L., Ambitwistor strings in four dimensions,
 \href{https://doi.org/10.1103/PhysRevLett.113.081602}{\textit{Phys. Rev. Lett.}} \textbf{113} (2014), 081602, 5~pages,
 \href{https://arxiv.org/abs/1404.6219}{arXiv:1404.6219}.

\bibitem{2021HHTZ}
He S., Hou L., Tian J., Zhang Y., Kinematic numerators from the worldsheet:
 cubic trees from labelled trees, \href{https://doi.org/10.1007/JHEP08(2021)118}{\textit{J.~High Energy Phys.}} \textbf{2021}
 (2021), no.~8, 118, 24~pages, \href{https://arxiv.org/abs/2103.15810}{arXiv:2103.15810}.

\bibitem{2012H}
Hodges A., A simple formula for gravitational {MHV} amplitudes,
 \href{https://arxiv.org/abs/1204.1930}{arXiv:1204.1930}.

\bibitem{2012Kap}
Kapranov M., The geometry of scattering amplitudes, {T}alk at Banff Workshop
 ``The Geometry of Scattering Amplitudes'', 2012.

\bibitem{1985KLT}
Kawai H., Lewellen D.C., Tye S.-H.H., A relation between tree amplitudes of
 closed and open strings, \href{https://doi.org/10.1016/0550-3213(86)90362-7}{\textit{Nuclear Phys.~B}} \textbf{269} (1986), 1--23.

\bibitem{1988KK}
Kleiss R., Kuijf H., Multigluon cross sections and 5-jet production at hadron
 colliders, \href{https://doi.org/10.1016/0550-3213(89)90574-9}{\textit{Nuclear Phys.~B}} \textbf{312} (1989), 616--644.

\bibitem{2020Mafra}
Mafra C.R., Private correspondence, 2020.

\bibitem{2014MS}
Mafra C.R., Schlotterer O., Multiparticle {SYM} equations of motion and pure
 spinor {BRST} blocks, \href{https://doi.org/10.1007/JHEP07(2014)153}{\textit{J.~High Energy Phys.}} \textbf{2014} (2014),
 no.~7, 153, 40~pages, \href{https://arxiv.org/abs/1404.4986}{arXiv:1404.4986}.

\bibitem{2015MS}
Mafra C.R., Schlotterer O., Berends--{G}iele recursions and the {BCJ} duality
 in superspace and components, \href{https://doi.org/10.1007/JHEP03(2016)097}{\textit{J.~High Energy Phys.}} \textbf{2016}
 (2016), no.~3, 097, 20~pages, \href{https://arxiv.org/abs/1510.08846}{arXiv:1510.08846}.

\bibitem{2020MS}
Mafra C.R., Schlotterer O., One-loop open-string integrals from differential
 equations: all-order {$\alpha'$}-expansions at {$n$} points, \href{https://doi.org/10.1007/jhep03(2020)007}{\textit{J.~High
 Energy Phys.}} \textbf{2020} (2020), no.~3, 007, 79~pages,
 \href{https://arxiv.org/abs/1908.10830}{arXiv:1908.10830}.

\bibitem{2009MS}
Mason L., Skinner D., Gravity, twistors and the {MHV} formalism, \href{https://doi.org/10.1007/s00220-009-0972-4}{\textit{Comm.
 Math. Phys.}} \textbf{294} (2010), 827--862, \href{https://arxiv.org/abs/0808.3907}{arXiv:0808.3907}.

\bibitem{2017Miz}
Mizera S., Combinatorics and topology of {K}awai--{L}ewellen--{T}ye relations,
 \href{https://doi.org/10.1007/jhep08(2017)097}{\textit{J.~High Energy Phys.}} \textbf{2017} (2017), no.~8, 097, 54~pages,
 \href{https://arxiv.org/abs/1706.08527}{arXiv:1706.08527}.

\bibitem{2018Miz}
Mizera S., Scattering amplitudes from intersection theory, \href{https://doi.org/10.1103/PhysRevLett.120.141602}{\textit{Phys. Rev.
 Lett.}} \textbf{120} (2018), 141602, 6~pages, \href{https://arxiv.org/abs/1711.00469}{arXiv:1711.00469}.

\bibitem{2019Miz}
Mizera S., Aspects of scattering amplitudes and moduli space localization,
 \textit{Springer Theses}, \href{https://doi.org/10.1007/978-3-030-53010-5}{Springer}, Cham, 2020, \href{https://arxiv.org/abs/1906.02099}{arXiv:1906.02099}.

\bibitem{2014MC}
Monteiro R., O'Connell D., The kinematic algebras from the scattering
 equations, \href{https://doi.org/10.1007/JHEP03(2014)110}{\textit{J.~High Energy Phys.}} \textbf{2014} (2014), no.~3, 110,
 28~pages, \href{https://arxiv.org/abs/1311.1151}{arXiv:1311.1151}.

\bibitem{1993R}
Reutenauer C., Free {L}ie algebras, \textit{London Mathematical Society
 Monographs. New Series}, Vol.~7, The Clarendon Press, Oxford University
 Press, New York, 1993.

\bibitem{2004RSV}
Roiban R., Spradlin M., Volovich A., Tree-level {$S$} matrix of {Y}ang--{M}ills
 theory, \href{https://doi.org/10.1103/PhysRevD.70.026009}{\textit{Phys. Rev.~D}} \textbf{70} (2004), 026009, 10~pages,
 \href{https://arxiv.org/abs/hep-th/0403190}{arXiv:hep-th/0403190}.

\bibitem{1991SV}
Schechtman V.V., Varchenko A.N., Arrangements of hyperplanes and {L}ie algebra
 homology, \href{https://doi.org/10.1007/BF01243909}{\textit{Invent. Math.}} \textbf{106} (1991), 139--194.

\bibitem{2002S}
Schocker M., Lie elements and {K}nuth relations, \href{https://doi.org/10.4153/CJM-2004-039-4}{\textit{Canad.~J. Math.}}
 \textbf{56} (2004), 871--882, \href{https://arxiv.org/abs/math.RA/0209327}{arXiv:math.RA/0209327}.

\bibitem{2010Stie}
Stieberger S., Open \& closed vs.\ pure open string disk amplitudes,
 \href{https://arxiv.org/abs/0907.2211}{arXiv:0907.2211}.

\bibitem{1973H}
't~Hooft G., A planar diagram theory for strong interactions, \href{https://doi.org/10.1016/0550-3213(74)90154-0}{\textit{Nuclear
 Phys.~B}} \textbf{72} (1974), 461--473.

\end{thebibliography}
\end{document}